%% file: CRv2.tex
\documentclass[conference,letterpaper]{IEEEtran}

\usepackage[utf8]{inputenc} 
\usepackage[T1]{fontenc}
\usepackage{url}
\usepackage{ifthen}
\usepackage{caption}
\usepackage{cite}
\usepackage{balance}
\usepackage{amssymb,graphicx}
\usepackage[cmex10]{amsmath} 
\usepackage{amsfonts,amsthm,amssymb}

\newtheorem{lemma}{Lemma}

\newtheorem{remark}{Remark}

\newtheorem{assumption}{Assumption}

\newtheorem{propo}{Proposition}

\input{defns}

\begin{document}
\title{Communication-Efficient Distributed Multiple Testing for Large-Scale Inference} 

 \author{%
   \IEEEauthorblockN{Mehrdad Pournaderi and Yu Xiang}
   \IEEEauthorblockA{University of Utah\\
                     50 Central Campus Dr 2110, Salt Lake City, USA\\
                     Email: \{mehrdad.pournaderi, yu.xiang\}@utah.edu}
 }

\maketitle
\begin{abstract}
 The Benjamini-Hochberg (BH) procedure is a celebrated method for multiple testing with false discovery rate (FDR) control. In this paper, we consider large-scale distributed networks where each node possesses a large number of p-values and the goal is to achieve the global BH performance in a communication-efficient manner. 
We propose that every node performs a local test with an adjusted test size according to the (estimated) global proportion of true null hypotheses. With suitable assumptions, our method is asymptotically equivalent to the global BH procedure. 
Motivated by this, we develop an algorithm for star networks where each node only needs to transmit an estimate of the (local) proportion of nulls and the (local) number of p-values to the center node; the center node then broadcasts a parameter (computed based on the global estimate and test size) to the local nodes. 
In the experiment section, we utilize existing estimators of the proportion of true nulls and consider various settings to evaluate the performance and robustness of our method.
%
 
 
\end{abstract}


\section{Introduction}

\label{sec:intro}
%
Motivated by microarray applications, the multiple testing problem has been the subject of active research in several fields. Since the formal introduction of the false discovery rate (FDR) control and the Benjamini-Hochberg (BH) procedure in~\cite{benjamini1995}, they have attracted significant interest in areas such as genomics and neuroimaging~\cite{genovese2002thresholding,abramovich2006adapting,efron2012large}, with many fundamental extensions~\cite{benjamini2001control,efron2001empirical,storey2002direct,sarkar2002some,genovese2002operating,efron2012large}.

In this work, we focus on distributed inference problems for FDR control based on the BH procedure, which is different from the existing literature on distributed detection and hypothesis testing formulations~\cite{tenney1981,tsitsiklis1984,viswanathan1997,blum1997}. Our work is largely motivated by a recent distributed FDR control formulation~\cite{Ramdas2017b} that leads to the Query-Test-Exchange (QuTE) algorithm with FDR control. In particular, the QuTE algorithm requires each node $i$ to transmit its p-values to all of its neighbor nodes $c\ge 1$ times to obtain p-values that are at nodes $c$ steps away (i.e., nodes connected by $c-1$ edges from node $i$). However, for a large-scale network with numerous nodes that each possesses, e.g., $10^6$ p-values, the amount of required communication quickly becomes infeasible or costly.
In fact, there is a trade-off between the amount of communication in the network and the statistical power of this algorithm. When no communication is allowed, each node performs the BH procedure (with a corrected test size) on its own p-values, leading to a Bonferroni type test; on the other hand, if unlimited communication is allowed in a connected network, at the end of the algorithm each node would have the entire set of p-values in the network, resulting in performing the global BH at each node. A natural question in a distributed network is whether it is possible to achieve FDR control with good power in a communication-efficient manner. 

To shed some light on this challenge, we take an asymptotic perspective and propose an aggregation method that is extremely communication-efficient. Our method is asymptotically equivalent to the global BH procedure, while reducing the communication cost to essentially two real-valued number per node. Even though the asymptotic analysis requires some assumptions, the proposed method is shown to be stable and robust in the finite sample regime according to our extensive simulation studies. A related attempt to reduce the communication cost of the QuTE algorithm has been made in~\cite{xiang2019distributed} but it still requires transmitting the quantized version of all the p-values in a network. Under a different measure (probability distribution of false alarms under FDR control), the authors in~\cite{ray2011false} study a distributed sensor network setting for detection of a single target in a region of interest, where each sensor node has only one p-value.

The rest of the paper is organized as follows. Section~II presents the background and problem formulation. In Section~III, we present our distributed BH algorithm along with the asymptotic analysis. We illustrate the performance of our method via a variety of simulations in Section~IV. 
%

%
\section{Background and Problem Settings}
\label{sec:background}

\subsection{Multiple Testing and False Discovery Rate Control}
\label{sec:FDR}


Consider testing the hypotheses $\mathsf{H_{0,k}}$ against their corresponding alternatives $\mathsf{H_{1,k}}$ for $1 \leq k\leq m$, according to the test statistics $X_k, 1 \leq k\leq m$, where $m_0$ of them are generated according to the null hypotheses.
Let $P_k = p(X_k)$, $1\le k\le m$, denote the p-values computed under the null hypothesis $\mathsf{H_{0,k}}$. The goal of multiple testing is to test the $m$ hypotheses while controlling a simultaneous measure of type I error. A rejection procedure controls a measure of error at some prefixed level $\alpha$, for $0<\a<1$, if it guarantees that the error is at most $\alpha$. The two major approaches for simultaneous error control are family-wise error rate (FWER) control and false discovery rate (FDR) control. FWER concerns with controlling the probability of making at least one false rejection, while FDR is a less stringent measure of error that aims to control the \emph{proportion} of false rejections among all rejections in expectation. For $m$ realizations of $p$-values, $(p_1,p_2,\ldots, p_m)$, denote the ordered $p$-values by  $(p_{(1)},p_{(2)},\ldots, p_{(m)})$, where $p_{(1)}$ and $p_{(m)}$ denote the smallest and largest $p$-values, respectively. Let $R$ and $V$ denote the number of rejections and false rejections by some procedure, respectively. Then, FWER is defined as $\mathbb{P}(V>0)$ and it can be controlled at some target level $\alpha$ by rejecting $\{k: p_k \leq \alpha/m\}$, known as Bonferroni correction. On the other hand, FDR is defined as
\begin{equation}
	\text{FDR} = \mathbb{E}\left(\frac{V}{\max\{R,1\}}\right),
\end{equation}
and the celebrated Benjamini-Hochberg (BH) procedure~\cite{benjamini1995} controls the $\text{FDR}$ at level $\a$ by rejecting the $\hat{k}$ smallest p-values, where $\kh = \max\big\{0\leq k\leq m: p_{(k)}\le \tau_k \big\}$ with $\tau_k=\a k/m$ and $p_{(0)}=0$. The rejection threshold is usually denoted by $\tau_{\text{BH}}: = \tau_{\hat{k}}$.

\subsection{Distributed False Discovery Control: Problem Setting}
\label{sec:distributed}
In the distributed scheme we are going to propose, our goal is to attain the global BH procedure performance, asymptotically. By global performance, we mean rejecting according to the centralized (pooled) BH threshold. As a result, we get the same FDR and power as the centralized BH procedure, where
\begin{equation}
    \text{power}=\mathbb{E}\left(\frac{R-V}{\max\{m_1,1\}}\right),
\end{equation}
and $m_1=m-m_0$ denotes the total number of true alternatives.
We follow the mixture model formulation~\cite{genovese2002operating}, where a null (or alternative) hypothesis is true with probability $r_0$ (or $r_1$). 

Consider a network consisting of $N$ nodes. Suppose the p-values in the network are generated $\iid$ according to the (mixture) distribution function $\Gc(t)=\sum_{i=1}^N{q^{(i)}\,G\big(t;r_0^{(i)}\big)}$, where a p-value is generated in node $i$ with probability $q^{(i)}$ and according to the distribution function $G\big(t;r_0^{(i)}\big) = r_0^{(i)}\,U(t) + (1-r_0^{(i)}) F(t),\ 0 \leq r_0^{(i)} < 1$, with $U$ denoting the CDF of p-values under their corresponding null hypotheses, (which is $\U[0, 1]$,) and $F$ denoting the common (but \emph{unknown}) distribution function of the p-values under $\mathsf{H_1}$. 
\begin{assumption}
Let $m$ denote the number of p-values in the network. For all nodes $i\in \{1,...,N\}$, we assume $r_0^{(i)}(m)$ is fixed, i.e., $r_0^{(i)}(m)=r_0^{(i)}<1$ for all $m\in\mathbb{N}$, and $q^{(i)}_m\rightarrow q^{(i)}$ as $m\rightarrow \infty$ where $\sum_{m=1}^\infty{q^{(i)}_m}=\infty$.\footnote{We make the assumption $\sum_{m=1}^\infty{q^{(i)}_m}=\infty$ to ensure $m^{(i)}\,\rightarrow\,\infty$ for all\ $1\leq i\leq N$ (with probability 1 as $m\to \infty$) by the Borel-Cantelli lemma. 
So this assumption can be interpreted as we ignore the nodes with finite number of p-values in the limit.  }\label{ass:fixed}
\end{assumption}
Under Assumption \ref{ass:fixed}, we have $\Gc_m(t) = G\big(t;r_0^*(m)\big)$, where $r_0^*(m)=\sum_{i=1}^{N}{q^{(i)}_m\,r_0^{(i)}}$ is the global probability of true nulls. Also, $r_0^*(m)\rightarrow r_0^*$, where $r_0^*=\sum_{i=1}^{N}{q^{(i)}\,r_0^{(i)}}$.
Let $\tau^*_{\text{BH}}(\alpha)$ and $\tau^{(i)}_{\text{BH}}(\alpha^{(i)})$ denote the global and local ($i$-th node) rejection thresholds for some test sizes $\alpha$ and $\alpha^{(i)}$, respectively.
Under suitable assumptions, Theorem 1 in \cite{genovese2002operating} argues that
\begin{align*}
    \tau^*_{\text{BH}}(\alpha)&\xrightarrow{\mathcal{P}}
    \tau(\alpha;r_0^*),\\
    \tau^{(i)}_{\text{BH}}(\alpha^{(i)})&\xrightarrow{\mathcal{P}} \tau(\alpha^{(i)};r_0^{(i)}),
\end{align*} 
where $\tau(\alpha;r_0):=\sup\big\{t:G(t;r_0)=(1/\alpha)\,t\big\}$.
On the other hand, we have
\begin{equation*}
    \sup\{t:G(t;r_0)=(1/\alpha)\,t\}=\sup\{t:F(t)=\beta(\alpha;r_0)\,t\}
\end{equation*}
for $0 \leq r_0<1$, where 
 \begin{align*}
 		\beta(\alpha;r_0) := \frac{(1/\a)-{r}_0}{1-{r}_0}\ .
 \end{align*}
Therefore, $\tau^*_{\text{BH}}$ admits a characterization via $F(t)$ and $\beta(\alpha;r_0^*)$ in the limit (and same for $\tau^{(i)}_{\text{BH}}$).
This (asymptotic) representation
leads to a key observation: Each node can leverage the global proportion of true nulls $r_0^*$ (provided to them) to achieve the global performance by calibrating its (local) test size $\alpha^{(i)}$ such that $\beta(\alpha^{(i)};r_0^{(i)})=\beta(\alpha;r_0^*)=:\beta^*$ 
to reach the global threshold. 
Specifically, it is straightforward to observe that setting $\a^{(i)} := \big((1-{r}_0^{(i)})\beta(\alpha;r_0^*) + {r}_0^{(i)}\big)^{-1}$ in $\beta(\alpha^{(i)};r_0^{(i)})$ at node $i$ will result in the global performance asymptotically.
We will formally present our \emph{distributed BH method} in the next section.

\section{Communication-Efficient Distributed BH}
\label{sec:quantized-FDR}
Consider a network consisting of one center node and $N$ other nodes, where the center node collects (or broadcasts) information from (or to) all the other $N$ nodes. Each node $i$, $i\in \{1,...,N\}$, owns $m^{(i)}$ 
p-values $\Pv^{(i)}=(P_1^{(i)},..., P_{m^{(i)}}^{(i)})$, where $m^{(i)}_0$ (or $m^{(i)}_1$) of them correspond to the true null (or alternative) hypotheses with $m^{(i)}_0 + m^{(i)}_1 = m^{(i)}$. 
Let $\hat{r}^{(i)}_0$ denote an estimator of $r^{(i)}_0$. In particular, in the simulation section, we will focus on two estimators of the upper bound of $r^{(i)}_0$~\cite{storey2002direct,swanepoel1999limiting}. 
\subsection{Algorithm}
For an overall targeted FDR level $\alpha$, our \emph{distributed BH method} consists of three main steps. 
\begin{itemize}\setlength\itemsep{0.5em}
	\item[(1)] {\bf Collect p-values counts}: Each node~$i$ estimates $\hat{r}^{(i)}_0$ and then sends $\big(m^{(i)}, \hat{r}^{(i)}_0\big)$ to the center node.
	\item[(2)] {\bf Estimate global slope ($\b^*$)}: Based on $\big(m^{(i)}, \hat{r}^{(i)}_0\big)$, $i\in\{1,...,N\}$, the center node computes $\hat\beta^*$ and broadcasts it to all the nodes, where
	\begin{align*}
		\hat\beta^* = \frac{(1/\a)-\hat{r}^*_0}{1-\hat{r}^*_0}\geq 1,
	\end{align*}
	 with $\hat{r}^*_0= \frac{1}{m}\sum_{i=1}^N \hat{r}_0^{(i)}m^{(i)}$ and $m=\sum_{i=1}^N m^{(i)}$.
	\item[(3)] {\bf Perform BH locally}: Upon receiving $\hat\b^*$, each node computes its own $\hat\a^{(i)}$ and performs the BH procedure according to $\hat\a^{(i)}$, where
	\begin{align*}
		\hat\a^{(i)} = \frac{1}{\big(1-\hat{r}_0^{(i)}\big)\hat\b^* + \hat{r}_0^{(i)}}\leq 1.
	\end{align*}
\end{itemize}

\begin{remark}
We note that in the case of having only one node in the network ($N=1$), we get $\hat{r}^*_0= \hat{r}_0^{(1)}$ and the algorithm reduces to the BH procedure, i.e., $\hat\a^{(1)}=\alpha$.
\end{remark}

\subsection{Estimation of $r_0^{(i)}$}
To begin with, we briefly review the two estimators of $r_0^{(i)}$ we use in this paper. Let $P_{(1)}<...< P_{(m)}$ denote the ordered p-values at some node, generated \iid from $G(t;r_0)$ (defined in \ref{sec:distributed}) and $\mathsf{G}_m(t)$ denote the empirical CDF of $(P_1,..., P_m)$. The problem of estimating the ratio of alternatives, $r_1$, has been studied extensively~\cite{hochberg1990more,hengartner1995finite,swanepoel1999limiting,benjamini2000adaptive,efron2001empirical,storey2002direct}. We adopt the following two estimators that are \emph{strongly consistent} estimators of some upper bounds of $r_0$, as suggested in~\cite{genovese2004stochastic}: 

\smallskip
\noindent{{\bf Storey's estimator}}~\cite{storey2002direct}: for any $\l\in (0,1)$, 
	\begin{align*}
		\hat{r}^{\text{Storey}}_0 = \min\biggl\{\frac{1-\mathsf{G}_m(\l)}{1-\l}, \,1\biggr\};
	\end{align*} 

\noindent{{\bf Spacing estimator}}~\cite{swanepoel1999limiting}: 
\begin{equation*}
		\hat{r}^{\text{spacing}}_0 = \min\biggl\{\frac{2r_m}{m V_m},1\biggr\},
	\end{equation*}
where $r_m = m^{4/5}(\log(m)^{-2l})$, for any $l>0$, and 
	\begin{align*}
		V_m = \max_{r_m+1\le j\le m-r_m} (P_{(j+r_m)}-P_{(j-r_m)}).
	\end{align*}

In fact, both estimators converge almost surely to upper bounds of $r_0$ (known to be close to $r_0$~\cite{genovese2004stochastic}).
This follows from some simple algebra combined with $\mathsf{G}_m(\l)\xrightarrow{a.s.} G(\lambda;r_0)$ by the strong law of large numbers for Storey's estimator; and from $(2r_m)/(m V_m)\xrightarrow{a.s.} \min_{0<t<1} (r_0+r_1 f(t))\geq r_0$ for the spacing estimator, where $f(t)$ denotes the density function of $F(t)$. Under slightly different conditions, another estimator is proposed in~\cite{hengartner1995finite}. We skip it in this work due to its similar asymptotic behavior to the spacing estimator. 
\subsection{Asymptotic Equivalence to the Global BH}
\begin{lemma}
Under Assumption \ref{ass:fixed}, if $\hat{r}_0^{(i)}\xrightarrow{a.s.}\overline{r}_0^{\,(i)}$ as $m^{(i)}\rightarrow\infty$, for all $1\leq i\leq N$, then we have 
\begin{subequations}
\begin{equation}
    \hat{r}^*_0\xrightarrow{a.s.}\overline{r}^*_0\ ,\quad
    \hat\beta^*\xrightarrow{a.s.}\overline{\beta}^*\ ,\quad
    \hat\a^{(i)}\xrightarrow{a.s.}\widetilde{\alpha}^{(i)}\ ,
\end{equation}
\end{subequations}
as $m\rightarrow\infty$, where
\begin{subequations}
\begin{align}
   \overline{r}^*_0 &=\sum_{i=1}^{N}{\overline{r}_0^{\,(i)}q^{(i)}}\ ,\\
    \overline{\beta}^* &= \frac{({1}/{\alpha})-\overline{r}^*_0}{1-\overline{r}^*_0} \ ,\\
    \widetilde{\alpha}^{(i)} &= \Big({\big(1-\overline{r}_0^{\,(i)}\big)\overline{\beta}^*+\overline{r}_0^{\,(i)}}\Big)^{-1}\ .
\end{align}
\end{subequations}
\end{lemma}
\begin{proof}
Since $\sum_{m=1}^\infty{q^{(i)}_m}=\infty$ (and according to the independence of p-values), we have $m^{(i)}\rightarrow\infty$ with probability 1 for all $1\leq i\leq N$ as $m\rightarrow\infty$. We note that $\sum_{k=1}^\infty {{k^{-2}}{q^{(i)}_k\big(1-q^{(i)}_k\big)}}<\infty$ and ${\lim}_{m\rightarrow\infty}\big(\frac{1}{m}\sum_{k=1}^m{q^{(i)}_k}\big)=q^{(i)}$ (recall that $\lim_{m\to \infty}q^{(i)}(m)=q^{(i)}$). Therefore,  $m^{(i)}/m\xrightarrow{a.s.}q^{(i)}$ by Kolmogorov’s strong law of large numbers. Direct application of the continuous mapping theorem proves the claims.
\end{proof}

\begin{assumption}
To simplify the technical arguments, we assume that the estimators of $r_0^{(i)}$ are consistent, i.e., $\overline{r}_0^{\,(i)}=r_0^{\,(i)}$ for all $1\leq i\leq N$. \label{ass:consis}
\end{assumption}
We note that under Assumption \ref{ass:consis}, we have $\overline{r}_0^*=r_0^*$, $\overline{\beta}^*=\beta^*$, and $\widetilde{\alpha}^{(i)}={\alpha}^{(i)}$.
Recall that $F$ denotes the (common) CDF of the p-values under $\mathsf{H_1}$ and define\footnote{In general, $\{t:F(t)={\beta}^*\, t\}\neq \varnothing$ and the solutions of $F(t)=\beta\,t$ are bounded by $1/\beta \leq \alpha$. Hence, the supremum always exists.}
\begin{equation*}
    \tau^*:=\tau(\alpha;r_0^*)=\sup\{t:F(t)={\beta}^*\, t\} ,
\end{equation*}
where $\beta^*=\beta(\alpha;r_0^*)$.
\begin{assumption}
$F(t)$ is continuously differentiable at $t={\tau}^*$ and $F'({\tau}^*)\neq {\beta}^*$.
\label{concav}
\end{assumption}
\begin{remark}
It should be noted that $F$ can be a mixture (or compound) distribution \cite[Theorem 7]{genovese2002operating}. 
\end{remark}
   We note that $\beta(\alpha^{(i)};r_0^{(i)})=\beta^*$ by the definition of $\alpha^{(i)}$ (in section \ref{sec:distributed}) and as a result, $\tau(\alpha^{(i)};r_0^{(i)})=\tau^*$. Therefore, according to Lemma \ref{het} in the appendix, Assumptions \ref{ass:fixed}, \ref{ass:consis}, and \ref{concav} are sufficient to imply $\tau_{\text{BH}}^{(i)}({\alpha}^{(i)})\xrightarrow{\Pc} \tau^*$ (\emph{exponentially fast}) as $m\to\infty$, where $\tau_{\text{BH}}^{(i)}({\alpha}^{(i)})$ denotes the BH rejection threshold at node $i$ with the target FDR ${\alpha}^{(i)}$. 
The following lemma concerns the asymptotic validity of this argument for the BH procedure based on the plug-in estimator $\hat{\alpha}^{(i)}$, i.e., we wish to show that $\tau_{\text{BH}}^{(i)}(\hat{\alpha}^{(i)})\xrightarrow{\Pc} {\tau}^*$ holds as well. To simplify the notation, we drop the superscript $(i)$ in the following lemma. 
\begin{lemma}
Under Assumptions \ref{ass:fixed}, \ref{ass:consis}, and \ref{concav}, if $\hat{\alpha}_m\xrightarrow{a.s.}\alpha$ as $m\to\infty$, then $ \tau_{\text{BH}}(\hat\alpha_m)\xrightarrow{a.s.}\tau_\alpha:=\tau(\alpha;r_0)$, where,
\begin{equation}
    \tau_{\text{BH}}(\hat\alpha_m) = \frac{\hat\alpha_m}{m}\max\big\{0\leq k\leq m:P_{(k)}\leq (k/m)\hat\alpha_m\big\}\  \nonumber
\end{equation}
(with $P_{(0)} = 0$), is the rejection threshold for the BH procedure with estimated target FDR $\hat\alpha_m$. \label{thm:plugin}
\end{lemma}
\begin{proof}
According to Assumption \ref{concav}, there exist a $\delta$-neighborhood of $\alpha$, $\mathcal{B}_{\delta}(\alpha)$, such that $F'(\tau_{\breve{\alpha}})\neq \beta_{\breve{\alpha}}$ for all $\breve{\alpha}\in\mathcal{B}_{\delta}(\alpha)$.
Fix some $0<\delta'< \delta$. By the almost sure convergence of $\hat\alpha_m$, we have $|\hat\alpha_m-\alpha|\leq\delta'$ for $m > \widetilde{m}(\omega)$. Hence, for large $m$ we get 
\begin{equation*}
    \tau_{\text{BH}}(\alpha-\delta')\leq \tau_{\text{BH}}(\hat\alpha_m)\leq \tau_{\text{BH}}(\alpha+\delta') \quad \ a.s.
\end{equation*}
Therefore, according to Lemma \ref{het} (with fixed $r_0(m)$) we get 
\begin{equation*}
    (1-\eps)\tau_{\alpha-\delta'}\leq \tau_{\text{BH}}(\hat\alpha_m)\leq (1+\eps)\tau_{\alpha+\delta'} \ a.s.,\ \text{large $m$}
\end{equation*}
for all $\eps>0$ and $0<\delta'< \delta$, completing the proof. 
\end{proof}
\begin{propo}
Under Assumptions \ref{ass:fixed}, \ref{ass:consis}, \ref{concav}, if $\tau^*>0$ then our distributed BH method attains the global performance (i.e., centralized FDR and power) as $m\,\rightarrow\,\infty$ where $m$ denotes the number of p-values in the network.
\end{propo}
\begin{proof}
By Lemma \ref{thm:plugin}, we have $\tau_{\text{BH}}^{(i)}(\hat{\alpha}^{(i)})\xrightarrow{\Pc} {\tau}^*$. Also, we have $\tau^*_{\text{BH}}(\alpha)\xrightarrow{\mathcal{P}}\tau^*$ according to Lemma \ref{het}. 
The convergence of FDR and power follows from the convergence (in probability) and boundedness of $\text{FDP}=\frac{V}{R\vee 1}$ and $\text{TDP}=\frac{R-V}{m_1\vee 1}$, respectively.
\end{proof}

\section{Simulations}
In this section, we demonstrate our proposed algorithm in various settings. In all the experiments, we set $\alpha = 0.2$ and number of nodes $N=50$. The \emph{estimated} FDR and power are computed by averaging over $200$ trials. 

The samples are distributed according to $\mathcal{N}(0,1)$ under $\mathsf{H_0}$. Under $\mathsf{H_1}$, 
we consider mixture alternatives (i.e., $\mathcal{N}(\mu,1)$ with random $\mu$ and we fix the distribution of $\mu$ for all nodes) and composite alternatives (i.e., we generate samples according to $\mathcal{N}(\mu^{(i)},1)$ at node $i$ with a unique distribution function for $\mu^{(i)}$).
As a reference, we perform the global (referred to as \emph{central} in the plots below) multiple testing by carrying out the BH procedure over all the p-values from all nodes, i.e., $\{\Pv^{1},..., \Pv^{N}\}$. We fix the hyper-parameters $\l=0.5$ for Storey's estimator and $l=0.5$ for the spacing estimator in all our simulations. The empirical performance of the spacing estimator is more stable in comparison with that of Storey's estimator, but the hyper-parameters can potentially be optimized (e.g., one can select $\l$ to minimize the mean-square error of the estimator via bootstrapping~\cite[Section 9]{storey2002direct}). 
\begin{figure}[!h]
\centering
  \includegraphics[scale=0.31]{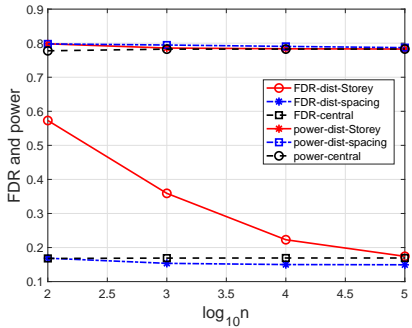}
  \caption{Experiment~1 (Each node owns $m^{(i)}=n$ p-values.).}\vspace{-1em}
\end{figure}

\smallskip
\noindent{{\bf Experiment 1 (Same number of p-values for all nodes).} In this experiment, we set $m^{(i)}=n$ and vary it from $10^2$ to $10^5$. To focus more on the sparse setting (i.e., when $r_1$ is small), we randomly pick $r_1^{(i)}\sim \U[0,0.3]$, and compute the number of alternatives $m^{(i)}_1=\lfloor r_1^{(i)}n\rfloor$. To ensure reproducibility, we fix $r_1^{(i)}=0.3*(i/N)$ for generating the p-values. Then we run $200$ trials, where in each trial, we fix $\mu_{\text{base}}=3$, and generate samples under $\mathsf{H_1}$ by first picking $\mu\sim\U\{[-\mu_{\text{base}}-0.5, -\mu_{\text{base}}+0.5]\cup [\mu_{\text{base}}-0.5, \mu_{\text{base}}+0.5]\}$ and then generating $X^{(i)}_j\sim \mathcal{N}(\mu,1)$.}

\smallskip
\noindent{{\bf Experiment 2 (Different number of p-values for each node).} In this experiment, we vary $n$ from $10^2$ to $10^6$ and randomly sample $m^{(i)}= n^{0.2+0.8A_i}$, where $A_i\sim\U[0,1]$. To ensure reproducibility, we fix $m^{(i)}= n^{0.2+0.8*(i/N)}$ for generating the p-values. The setting is otherwise the same as in Experiment~1.}

\begin{figure}[!h]
\centering
  \includegraphics[scale=0.31]{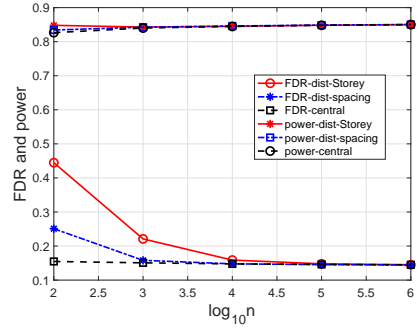}
  \caption{Experiment~2 (Different $m^{(i)}$ for each node).}\vspace{-1em}
\end{figure}

\smallskip
\noindent{{\bf Experiment 3 (Vary $\mu$).} In this experiment, we consider the opposite case. We fix $m^{(i)}$'s by setting $n=10^4$ and pick $m^{(i)}= n^{0.2+0.8*(i/N)}$. We vary $\mu_{\text{base}}$ from $2$ to $5$, and generate samples under $\mathsf{H_1}$ in the same way as in Experiment~1.}

\begin{figure}[!h]
\centering
  \includegraphics[scale=0.31]{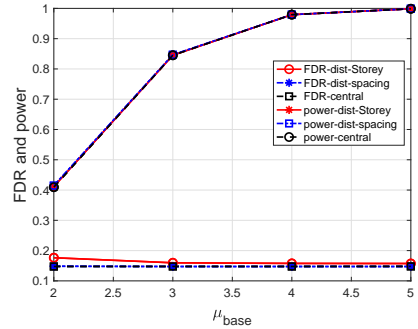}
  \caption{Experiment~3 (Vary $\mu$).}\vspace{-1em}
\end{figure}

\smallskip
\noindent{{\bf Experiment 4 (Heterogeneous alternatives).} We consider a setting where each node~$i$ generates samples according to a unique $\mu^{(i)}_{\text{base}}$ under $\mathsf{H_1}$. Specifically, we fix $\mu_{\text{base}}^{(i)}=2+i/N$ and pick $\mu^{(i)}\sim\U\{[-\mu_{\text{base}}^{(i)}-0.5, -\mu_{\text{base}}^{(i)}+0.5]\cup[\mu_{\text{base}}^{(i)}-0.5, \mu_{\text{base}}^{(i)}+0.5]\}$. We set $m^{(i)}= n^{0.2+0.8*(i/N)}$. And then we fix $r_1^{(i)}$, and vary $n$ from $10^3$ to $10^6$ to generate samples under $\mathsf{H_1}$ in the same way as in Experiment~1.}

\begin{figure}[!h]
\centering
  \includegraphics[scale=0.31]{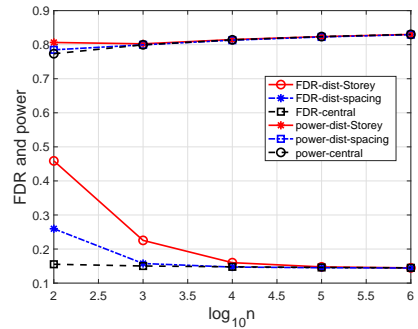}
  \caption{Experiment~4 (Heterogeneous alternatives).}\vspace{-1em}
\end{figure}

\smallskip

\noindent{{\bf Experiment 5 (Dependent p-values).} Finally, we evaluate the robustness of our method by considering dependent p-values. In particular, we adopt two commonly used covariance structures: (I) $\Sigma_{i,j} = \rho^{|i-j|}$, and (II) $\Sigma_{i,i} = 1$, $\Sigma_{i,j} = \rho\cdot\ind (\lceil i/20\rceil=\lceil j/20\rceil)$, where $\ind(\cdot)$ denotes the indicator function. We vary $\rho$ from $0$ to $0.8$ and fix $n=10^3$ for $m^{(i)}= n^{0.2+0.8*(i/N)}$ and $\mu_{\text{base}}=3$. The setting is otherwise the same as in Experiment~1.}

\begin{figure}[!h]
\centering
  \includegraphics[scale=0.31]{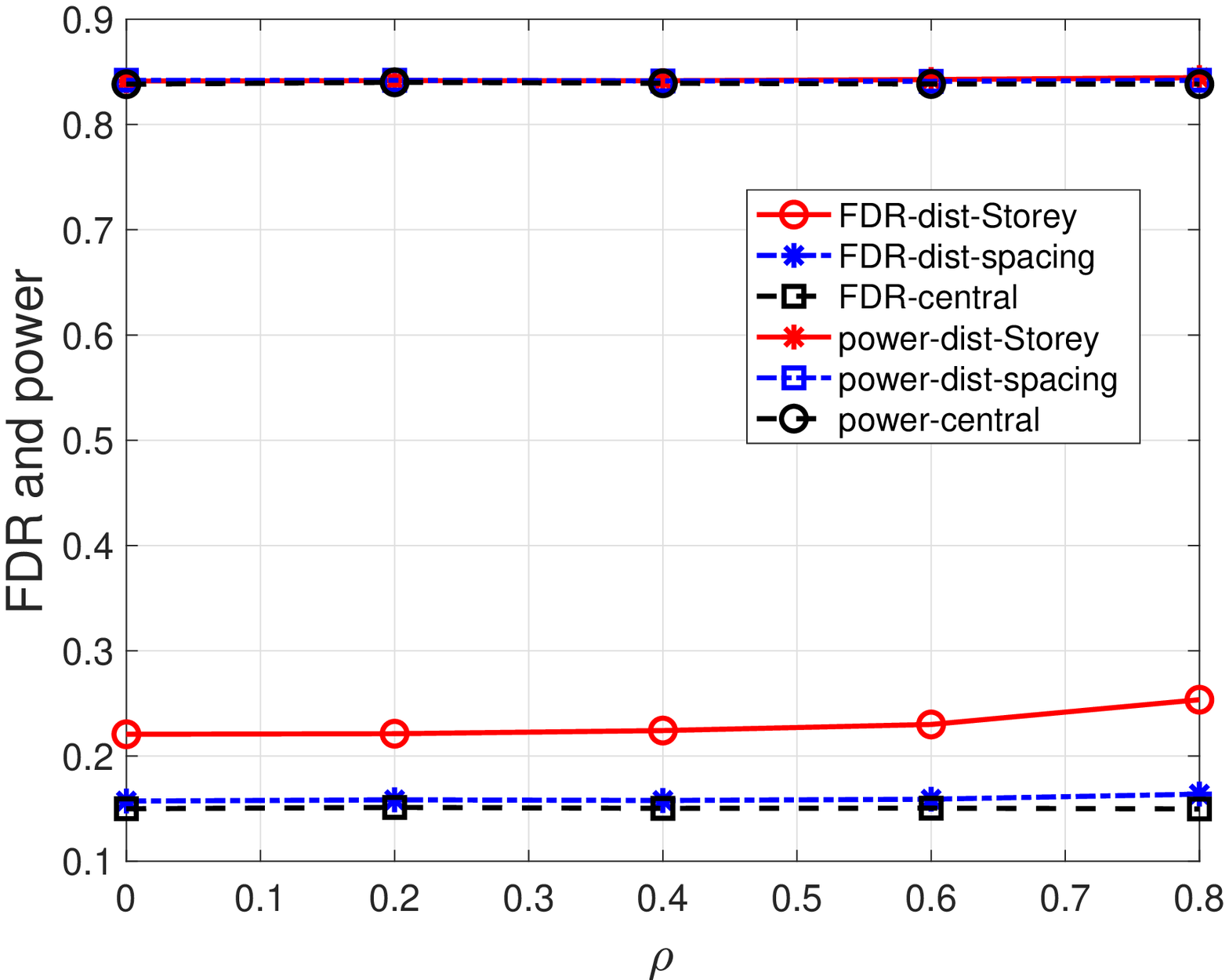}
  \caption{Experiment~5 with covariance structure (I).}\vspace{-1em}
\end{figure}

\begin{figure}[!h]
\centering
  \includegraphics[scale=0.31]{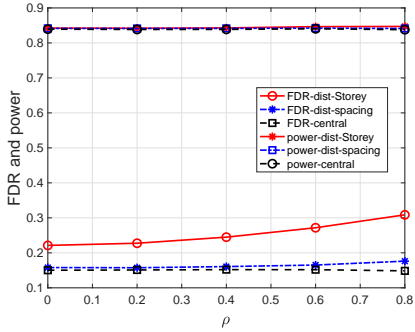}
  \caption{Experiment~5 with covariance structure (II).}\vspace{-1.3em}
\end{figure}

\section{Discussion}
In this work, we have initiated a methodology for distributed large-scale multiple testing which is based on (one-shot) aggregation of proportions of the true nulls at a center node. Our simulations with Gaussian statistics show that the method is robust to deviations from the assumptions. 
Considering the minimal communication budget we allow, the only potential competitor is the no-communication (or no-aggregation) method, which has the asymptotic FDR control property.
We believe that our algorithm can improve significantly upon the power of the no-communication method (in both local and global senses) for some challenging cases. We leave these comparisons for the extended version of this work. 
\appendix

Consider performing the BH procedure on $m$ p-values generated $\iid$ according to $G(t;r_0)$ (defined in Section II-B). The following technical lemma concerns relaxing the assumptions of the  Theorem~1 in~\cite{genovese2002operating} in two directions: (I) $F$ is not assumed to be concave and multiple solutions to $F(t) = \beta(\alpha;r_0)\, t$ can exist, (II) $r_0(m)$ is not assumed to be fixed for each $m$. 


\begin{lemma}
Suppose $r_0(m)\rightarrow r_0$ and fix $\alpha\in(0,1)$. Define $\beta_{r_0}:=\beta(\alpha;r_0)$ and $\tau_{r_0}:=\tau(\alpha; r_0)=\sup\{t:F(t) = \beta_{r_0}\, t\}$. If $F(t)$ is continuously differentiable at $t=\tau_{r_0}$ and $F'(\tau_{r_0})\neq {\beta}(\alpha;r_0)$, then 
$\tau_{\text{BH}}(\alpha;r_0(m))\xrightarrow{a.s.} \tau_{r_0}$. \label{het}
\end{lemma}
\begin{proof}
We follow the approach in \cite{genovese2002operating} and highlight the main differences. According to the continuous differentiability of $F(t)$ at $t=\tau_{r_0}$, there exist an open $\delta$-neighborhood of $r_0$, $\mathcal{B}_{\delta}(r_0)$ such that $F'({\tau}_{\breve{r}_0})\neq {\beta}_{\breve{r}_0}$ for all $\breve{r}_0\in\mathcal{B}_{\delta}(r_0)$.
Take some $0 < \delta' <\delta$ and pick $m'$ such that $|r_0(m)-r_0|\leq\delta'$ for $m\geq m'$. Let $\eps_m=(\log m)^{-1}$, $m\geq 3$.
Define
\begin{equation*}
    a_m = \frac{m\tau_{r_0+\delta'}}{\alpha}(1-\eps_m)\quad\text{and}\quad b_m = \frac{m\tau_{r_0-\delta'}}{\alpha}(1+\eps_m)\ .
\end{equation*}
Let $D_{\text{BH}} = m/\alpha\, \tau_{\text{BH}}$ denote the BH deciding index, or equivalently, the number of rejections made by the BH procedure. First, we show $\mathbb{P}(D_{\text{BH}} > b_m)\rightarrow 0$ as follows,
\begin{align}
    &\mathbb{P}(D_{\text{BH}} > b_m) \\
    &= \mathbb{P}\Bigg(\underset{k\,>\,b_m}{\bigcup}\big\{P_{(k)}\leq (k/m)\alpha\big\}\Bigg) \\
    &= \mathbb{P}\Bigg(\underset{k\,>\,b_m}{\bigcup}\bigg\{\sum_{i=1}^m \ind\big\{P_i\leq (k/m)\alpha\big\}\geq k\bigg\}\Bigg)\nonumber\\
    &\leq \sum_{k\,>\,b_m}\mathbb{P}\Bigg[\bigg(\sum_{i=1}^m \ind\big\{P_i\leq (k/m)\alpha\big\}\bigg)-\mu(k)\geq k-\mu(k)\Bigg],\nonumber
\end{align}
where 
\begin{align}
  \mu(k) & =  \mathbb{E}\Big(\sum_{i=1}^m \ind\big\{P_i\leq (k/m)\alpha\big\}\Big)\\
  & =r_0(m)\,m\,(k/m)\alpha + (1-r_0(m))\,m F\big((k/m)\alpha\big)\ .\nonumber
\end{align}
We observe that
\begin{align}\label{diff}
     k-&\mu(k) = m r_1(m)\Big[\frac{\alpha k}{m}\beta_{r_0(m)}-F\big(\frac{\alpha k}{m}\big)\Big].
\end{align}
We note that $\beta_{r_0-\d'}\,t-F(t)>0$ for all $t>\tau_{r_0-\delta'}$ according to the definition of $\tau_{r_0-\delta'}$.
Therefore, 
\begin{equation*}
    k-\mu(k)\geq m (r_1-\delta')\Big[\frac{\alpha k}{m}\beta_{r_0-\d'}-F\big(\frac{\alpha k}{m}\big)\Big]=:h(k)>0
\end{equation*}
for $k\geq b_m$ and large $m$.
Also, according to Taylor's theorem we have
\begin{equation}
    h(b_m)= m(r_1-\delta')\Big[\tau_{r_0-\d'}\,\eps_m\big(\beta_{r_0-\d'}-F'(\tau_{r_0-\d'})\big)+o(\eps_m)\Big]\nonumber \ .
\end{equation}
We observe,
\begin{align*}
    \underset{k > b_m}{\inf}h(k)/m &=  (r_1-\delta')\underset{k:\frac{\alpha k}{m}>\frac{\alpha b_m}{m}}{\inf} \Big[\frac{\alpha k}{m}\beta_{r_0-\d'}-F\big(\frac{\alpha k}{m}\big)\Big]\\
    &\geq (r_1-\delta')\underset{t>\tau_{r_0-\delta'}(1+\eps_m)}{\inf} \big(\beta_{r_0-\d'}t-F(t)\big).
\end{align*}
Recall that $\beta_{r_0-\d'}\,t-F(t)>0$ for all $t>\tau_{r_0-\delta'}$. But,
\begin{align*}
    (r_1-\delta')\Big[\beta_{r_0-\d'}\tau_{r_0-\delta'}(1+\eps_m)&-F(\tau_{r_0-\delta'}(1+\eps_m))\Big]\\
    &=\frac{h(b_m)}{m}=o(1).
\end{align*}
Hence, for large enough $m$ we get 
\begin{equation*}
    \underset{k > b_m}{\inf}h(k)/m\geq {h(b_m)}/{m}>0,
\end{equation*}
and as a result $h(k)\geq h(b_m)$ for all $k>b_m$.
Since $F'(\tau_{r_0-\d'})<\beta_{r_0-\d'}$ and $\eps_m=1/\log m$, we get $h(b_m)/\sqrt{m}\asymp \sqrt{m}/\log m$.
Hence, by Hoeffding's inequality we get
\begin{align}
    \mathbb{P}(D_{\text{BH}} > b_m) 
    &\leq \sum_{k\,>\,b_m}e^{-{2h(k)^2}/{m}}
    \precsim m\, e^{-{2h(b_m)^2}/{m}}\nonumber\\
    &\asymp m\, e^{-cm/{(\log m)}^2}\rightarrow 0.\label{UB1}
\end{align}
for some constant $c>0$.
Similarly, we note
\begin{align*}
    \mu &(a_m)-a_m \geq m (r_1-\delta')\Big[F\big(\frac{\alpha a_m}{m}\big)-\frac{\alpha a_m}{m}\beta_{r_0+\d'}\Big]\\
    &= m(r_1-\delta')\Big[\tau_{r_0+\d'}\,\eps_m\big(\beta_{r_0+\d'}-F'(\tau_{r_0+\d'})\big)+o(\eps_m)\Big]
\end{align*}
for large $m$.
Since $F'(\tau_{r_0+\d'})<\beta_{r_0+\d'}$, we get $m^{-1/2}\big(\mu\,(a_m)-a_m\big)\rightarrow\infty$ as $m\rightarrow\infty$. Hence,
\begin{align}
    \mathbb{P}(D_{\text{BH}} < a_m) 
    &\leq \mathbb{P}\Big(P_{(a_m)}> (a_m/m)\alpha\Big)\nonumber \\
    &\leq \mathbb{P}\bigg(\sum_{i=1}^m \ind\big\{P_i\leq (a_m/m)\alpha\big\}< a_m\bigg)\nonumber\\
    &\overset{(a)}{\leq} \exp\bigg(-\frac{2(a_m-\mu(a_m))^2}{m}\bigg)\nonumber\\
    &\precsim\, e^{-c' m/{(\log m)}^2}\label{UB2}
    \rightarrow 0\ ,
\end{align}
for some constant $c'>0$, where $(a)$ follows from Hoeffding's inequality. Since the upper bounds in \eqref{UB1} and \eqref{UB2} are summable in $m$, we get $a_m\leq D_\text{BH}\leq b_m,\ a.s.$ for all $0<\d'<\d$ and large $m$ by the Borel-Cantelli lemma, completing the proof. 
\end{proof}
\balance


\bibliographystyle{IEEEtran}
\bibliography{ref.bib}

\end{document}

%% file: defns.tex
\usepackage{xspace}
\usepackage{bbm}
\usepackage{mathrsfs}

\newcommand{\Gc}{\mathcal{G}}

\newcommand{\Pc}{\mathcal{P}}

\newcommand{\Pv}{{\bf P}}

\newcommand{\kh}{{\hat{k}}}

\def\a{\alpha}
\def\b{\beta}

\def\d{\delta}

\def\eps{\epsilon}
\def\l{\lambda}

\newcommand{\U}{\mathrm{Unif}}

\def\textiid{i.i.d.\@\xspace}
\newcommand\iid{\ifmmode\text{ i.i.d. } \else \textiid \fi}

\newcommand{\ind}{\boldsymbol{1}}

